\documentclass[5p]{elsarticle}
\usepackage[usenames]{color} 
\usepackage{graphicx}          
\usepackage{lineno,hyperref} 
\modulolinenumbers[5]
\usepackage[cmex10]{amsmath} 
\usepackage{amsfonts} 
\usepackage{amsthm}
\usepackage{mathtools} 
\usepackage{IEEEtrantools} 
\usepackage{color}   
\usepackage{subcaption}
\usepackage{stfloats}   
\usepackage{tikz}
\usetikzlibrary{arrows, decorations.pathmorphing, backgrounds, positioning,
fit, shapes.geometric} 
 
\usepackage{amscd}
\usepackage{amsmath}
\usepackage{amssymb}
\usepackage{graphicx}

\newtheorem {theo} {\bf Theorem}

\newtheorem {lem} [theo] {\bf Lemma}
\newtheorem {defn} [theo] {\bf Definition}

\newtheorem {rem} [theo] {\bf Remark}

\newtheorem {assum} {\bf Assumption}

\newcommand{\be}{\begin{eqnarray}}
\newcommand{\ee}{\end{eqnarray}}
\newcommand{\benn}{\begin{eqnarray*}}
\newcommand{\eenn}{\end{eqnarray*}}
\newcommand{\bse}{\begin{equation}}
\newcommand{\ese}{\end{equation}}
\newcommand{\bsenn}{\begin{displaymath}}
\newcommand{\esenn}{\end{displaymath}}
\newcommand{\logand}{\;\;{\rm and }\;\;}

\newcommand{\numVeh}{N}

\newcommand{\subvel}{v}		
\newcommand{\subpos}{x}		
\newcommand{\fric}{a} 		
\newcommand{\pos}{z}  		
\newcommand{\posdist}{x}	
\newcommand{\vel}{v}		
\newcommand{\contState}{c}	
\newcommand{\inp}{u}		
\newcommand{\regerr}{\epsilon}		
\newcommand{\gv}{g_v}		
\newcommand{\gx}{g_x}		
\newcommand{\asym}{\rho}
\newcommand{\rx}{\asym_\subpos}	
\newcommand{\rv}{\asym_\subvel}	
\newcommand{\bx}{\beta_\subpos}			
\newcommand{\bv}{\beta_\subvel}			

\newcommand{\wv}{c}		
\newcommand{\cp}{c_+}	
\newcommand{\cm}{c_-}	
\newcommand{\ampl}{A}	
\newcommand{\period}{T}	
\newcommand{\erl}{e}	
\newcommand{\areaosc}{E}	

\newcommand{\lapl}{L}		
\newcommand{\lv}{\lapl_\subvel} 	
\newcommand{\lx}{\lapl_\subpos}		
\newcommand{\laplc}{\hat{\lapl}}

\newcommand{\lvc}{\laplc_\subvel} 	
\newcommand{\lxc}{\laplc_\subpos}		
\newcommand{\systMat}{M_\numVeh}	
\newcommand{\systMatCirc}{\hat{M}_\numVeh} 
\newcommand{\eigVectM}{u}			

\newcommand{\eigVectL}{w}			
\newcommand{\sfr}{\phi}		
\newcommand{\eigl}{\lambda}			
\newcommand{\eiglx}{\eigl_\subpos}	
\newcommand{\eiglv}{\eigl_\subvel}	
\newcommand{\eigm}{\nu}				
\newcommand{\eigmcoef}{n}		

\usepackage{graphicx}

\begin{document}
\begin{frontmatter}

\title{Transients of platoons with asymmetric and different
Laplacians\tnoteref{support}} 
\tnotetext[support]{J.J.P. Veerman's research was partially supported by the European
Union's Seventh Framework Program (FP7-REGPOT-2012-2013-1) under grant agreement
n316165. I. Herman was supported by the Czech Science Foundation
within the project GACR 13-06894S. For the simulations we used the code
written by Carlos Cantos.}

\author[czechaddress]{Ivo Herman\corref{mycorrespondingauthor}}
\ead{ivo.herman@fel.cvut.cz}
\cortext[mycorrespondingauthor]{Corresponding author}
\author[czechaddress]{Dan Martinec}
\ead{dan.martinec@fel.cvut.cz}
\author[peteraddress1,peteraddress3]{J. J. P. Veerman}
\ead{veerman@pdx.edu}
\address[czechaddress]{Faculty of Electrical Engineering, Czech Technical
University in Prague. Czech Republic}
\address[peteraddress1]{Fariborz Maseeh Dept. of Math. and Stat., Portland State
Univ., Portland, OR, USA} 
\address[peteraddress3]{CCQCN, Dept of Physics, University of Crete, 71003
Heraklion, Greece}
%


\begin{abstract}
We consider an asymmetric control of platoons of identical vehicles with
nearest-neighbor interaction.
Recent results show that if the vehicle uses different asymmetries for position and
velocity errors, the platoon has a short transient and low
overshoots. In this paper we investigate the properties of vehicles with
friction. To achieve consensus, an integral part is added to the controller,
making the vehicle a third-order system. We show that the parameters can
be chosen so that the platoon behaves as a wave equation with different wave
velocities.
Simulations suggest that our system has a better performance than other nearest-neighbor scenarios.
Moreover, an optimization-based procedure is used to find the controller properties.
\end{abstract}
\end{frontmatter}
\begin{centering}
\vspace{-3pt}
\section{Introduction}
\label{chap:zero} 
\end{centering}
Control of vehicular platoons has become a field of intensive research. The
reason for this is the possibility to increase throughput and safety of the
highway traffic at the same time. Moreover, since the study of such large
systems gives asymptotic behavior and achievable limits, it is also appealing
from a theoretical perspective. 

It is well known that if some centralized information is present in the
system, then the performance is good and the system is scalable. These
approaches comprise either LQR control \cite{Bamieh2002} or
local control with added knowledge of the desired velocity \cite{Lin2012,
Barooah2009a}. Among the systems with permanent communication the Cooperative
Automatic Cruise Control is the most widely studied \cite{Milanes2014} and
implemented.

In this paper we consider a nearest-neighbor interaction without any
centralized information. The most important performance
measure is, apart from the settling time, a so called string stability. The
system is string stable if the disturbance acting at one vehicle does not
amplify as it propagates along the platoon \cite{Ploeg2014}. It is shown in
\cite{Middleton2010} that a time-headway spacing policy can achieve string stability. However, this policy increases the platoon length with the speed of
the platoon.

A common fixed-distance scenario is a symmetric bidirectional control, where the
vehicle measures the distance to both its neighbors in the platoon. The properties of such
control are investigated in \cite{Lin2012, Hao2012, 
Veerman2007}. One of the advantages is linear scaling of the
$\mathcal{H}_\infty$ norm with the number of vehicles \cite{Veerman2007}. The
main drawback is a very long transient for a higher number of vehicles and sensitivity to noise \cite{Bamieh2012}. Both can be qualitatively decreased using a feedback control of the leader \cite{Martinec2014b}. 

Asymmetric control was proposed to shorten the transient and increase the 
controllability of large platoons \cite{Barooah2009a,
Hao2012b} --- a vehicle uses different weights to its two neighbors.
Nevertheless, the price to pay is an exponential growth of $\mathcal{H}_\infty$ norm of the transfer functions with the graph distance between vehicles \cite{Tangerman2012,
Herman2013b}, so this approach is not scalable.

Both symmetric and asymmetric control share one common property
--- the asymmetry in coupling between between vehicles is the same for position and
velocity.
That is, only one graph Laplacian is used. Except for a simpler analysis, there
are not many reasons to limit ourselves to use only one type of coupling. As was
numerically shown in \cite{Hao2012c} and more thoroughly discussed in \cite{Cantos2014b}, the
response of a platoon with a symmetric coupling in position and asymmetric in
velocity can be scalable and have short transients. Thus, this approach
combines the advantages both from symmetric and asymmetric control.

The paper \cite{Cantos2014b} considers a double integrator model of the
system and uses wave properties for system analysis. In this paper we extend the
results of \cite{Cantos2014b} to a third-order system.
The main reason is the presence of friction in every real vehicle. We added an
integrator to the system to achieve coherent solutions. With such a system, the
stability and performance analysis of the platoon complicates.
Although a similar type of a system was
recently analyzed in \cite{Andreasson2014}, the results are not applicable for
our case since the coupling used in our work is non-symmetric.

Finally, to achieve a good performance, we provide an optimization procedure
based on the properties of waves in the platoon. This approach allows us to
optimize the controller parameters and weights of the communication graph at the
same time.

\section{Model of the vehicles} 
 \label{chap:second}
 
We assume $\numVeh+1$ identical vehicles travelling on a line, indexed as $0,
\ldots, \numVeh$. The first vehicle with index 0 is a leader which is driven
independently of the the rest of the formation. Unlike standard double integrator models
\cite{Bamieh2012,  Tangerman2012,Cantos2014b}, real systems have a friction, i.
e., there is a feedback from velocity, which eventually makes the vehicle to
stop.
The vehicle model is
\vspace{-3pt}
\begin{equation}
	\ddot {\posdist}_i = -\fric \vel_i + \inp_i, \label{eq:vehmod}
\end{equation}
where $\posdist_i$ is the position of the $i$th vehicle, $\vel_i=\dot
\posdist_i$ is its velocity, $\fric \in \mathbb{R}$ is the viscous friction
coefficient and $\inp_i$ is the input to the vehicle.

In order to enable the vehicles in the platoon to track the
leader moving with constant velocity, we need to satisfy the Internal Model
Principle \cite{Wieland2011, Lunze2012} which in our case means the presence of two
integrators in the open-loop model of each vehicle. Since one integrator is
already present in the vehicle model ($\dot \posdist_i=\vel_i$), it suffices to
add an integral action in the controller of each
vehicle. The controller is given as $\dot{\bar{\contState}}_i = \regerr_i$
with $\regerr_i$ defined in (\ref{eq:regerr}) and $\bar\contState_i$ is the
state of the integrator in the controller.
The input to the vehicle is then $\inp_i = \bar{\contState}_i$. 

Each vehicle uses only the information obtained from its nearest neighbors
--- the vehicle in front of it and behind. The goal of the vehicle is to keep a
prescribed spacing to them, i.e., $\posdist_{i-1}-\posdist_i \rightarrow d_{i-1,
i}$ with $d_{i,j}$ being the desired distance between $i$ and $j$. The
regulation error $\regerr_i$ comes from  the relative spacing and velocity errors as
\begin{IEEEeqnarray}{rCl}
	\regerr_i &=& \gx
	\big[(1-\rx)(\posdist_{i-1}-\posdist_i-d_{i-1,
	i})-\rx(\posdist_i-\posdist_{i+1} \label{eq:regerr}
	\\
	&&-d_{i, i+1})\big]
	+ \gv \big[(1-\rv)(\vel_{i-1} - \vel_i)-\rv(\vel_i
	-\vel_{i+1})\big],\nonumber
\end{IEEEeqnarray}
where the position asymmetry is labeled as $\rx$, velocity asymmetry as $\rv$
and the $\gx, \gv \in \mathbb{R}$ are weights of position and velocity errors.
The coupling in
position is symmetric if $\rx=0.5$ and asymmetric otherwise (the same for
$\rv$).

 To simplify the analysis, we introduce error variables
 $\pos_i=\posdist_i - \posdist_0 + d_{0, i}$. This implies
$\pos_{i-1}-\pos_i=\posdist_{i-1}-\posdist_i-d_{i-1, i}$ and $\dot \pos_{i-1}-\dot \pos _i=\dot \posdist_{i-1}-\dot \posdist_i$.
We impose that $d_{i,k}-d_{i, j}=d_{j,k}$ and $d_{i,j}=-d_{j,i}$.
The single vehicle model combined with the controller then has a form 
\begin{IEEEeqnarray}{rCl}
	\ddot \pos_i = -\fric \dot \pos_i + \bar{\contState},
	\quad
	\dot{\bar{\contState_i}} = \regerr_i.
\end{IEEEeqnarray}
We use a minor state transformation $\contState_i=\bar{\contState}_i-\fric
\vel_i$ to obtain a controller-canonical form of the individual-vehicle model
\begin{IEEEeqnarray}{rCl}
	\begin{pmatrix}
		\dot \pos_i \\ \ddot \pos_i \\ \dot \contState_i
	\end{pmatrix}
	= \begin{pmatrix}
		\dot \pos_i \\ \ddot \pos_i \\ \dddot \pos_i
	\end{pmatrix}
	= 
	\begin{pmatrix}
		0 & 1 & 0 \\
		0 & 0 & 1 \\
		0 & 0 & -\fric
	\end{pmatrix}
	\begin{pmatrix}
		\pos_i \\ \dot \pos_i \\ \ddot \pos_i
	\end{pmatrix}
	+
	\begin{pmatrix}
		0 \\ 0 \\ 1
	\end{pmatrix}
	\regerr_i.
\end{IEEEeqnarray}

 In a vector form we write the overall system of $N+1$ vehicles (including the
 leader) as
\begin{equation}
\dfrac{d}{dt}\begin{pmatrix} 
				\pos\\ 
				\dot \pos \\
 				\ddot \pos 
 			\end{pmatrix}  =
\systMat 	\begin{pmatrix} 
				\pos \\ 
				\dot \pos \\ 
				\ddot \pos 
			\end{pmatrix} \equiv
			\begin{pmatrix} 
				0 & I & 0 \\
				0 & 0 & I \\
				-\gx \lx & -\gv \lv & -\fric I \end{pmatrix}
\begin{pmatrix} \pos\\ \dot \pos \\ \ddot \pos \end{pmatrix},
\label{eq:overallSystemPath}
\end{equation}
where $\pos = [\pos_0, \ldots, \pos_\numVeh]^T$. 
Let us call the system
(\ref{eq:overallSystemPath}) \emph{a path system}, since the communication
topology is a weighted path graph. 
 The Laplacians $\lx, \lv \in \mathbb{R}^{\numVeh+1\times
\numVeh+1}$ of the path graph are defined as
{\setlength{\arraycolsep}{2pt}
\begin{IEEEeqnarray}{rCl} 
	\lx &=& 	\begin{pmatrix}
				0 & 0 & 0 & 0 & \ldots & 0 \\
				-(1-\rx) & 1 & -\rx & 0 & \ldots & 0 \\
				0 & -(1-\rx) & 1 & -\rx & \ldots & 0 \\
				\vdots & \vdots & \vdots & \vdots & \ddots & \vdots \\
				0 & 0 & 0 & \ldots & -1 & 1
	\end{pmatrix}, \label{eq:lx} \\
	\lv &=& 	\begin{pmatrix}
				0 & 0 & 0 & 0 & \ldots & 0 \\
				-(1-\rv) & 1 & -\rv & 0 & \ldots & 0 \\
				0 & -(1-\rv) & 1 & -\rv & \ldots & 0 \\
				\vdots & \vdots & \vdots & \vdots & \ddots & \vdots \\
				0 & 0 & 0 & \ldots & -1 & 1
			\end{pmatrix}. \label{eq:lv}
\end{IEEEeqnarray}
}

The last vehicle has no follower, so it uses only front spacing and velocity
errors. This type of boundary condition is called regular boundary condition
\cite{Cantos2014b}. The second boundary condition is that the leader is driven
independently of the platoon (zeros in the first rows of $\lx, \lv$).

We assume that initially the system in (\ref{eq:overallSystemPath}) is at
stand-still and then the leader starts to move with unit velocity:
\begin{IEEEeqnarray}{rCl}
	\pos_i(t)&=&0 \: \text{ for } i=1, \ldots, \numVeh, \: t < 0 ,
	\label{eq:initCond}\\
	\posdist_0(t)&=&0, \: t<0, \:\: \posdist_0(t) = t, \: t\geq0. \nonumber
\end{IEEEeqnarray}

\section{Relation to previous works}
Our paper builds on the results of works \cite{Cantos2014b, Cantos2014a}. Both papers deal with a signal
propagation in systems with nearest-neighbor interaction. The vehicle model is a
double integrator, i.e., $\ddot \posdist_i= \regerr_i$ with $\regerr_i$ given in
(\ref{eq:regerr}), so the papers use nearest-neighbor asymmetric
interaction. 

The work \cite{Cantos2014a} analyzes a system with a circular topology. We
call such system a \emph{circular system}. The interaction between the leader
and the vehicle $\numVeh$ is added, so the Laplacians $\lxc, \lvc \in
\mathbb{R}^{N+1\times N+1}$ of a circular graph are
{\setlength{\arraycolsep}{2pt}
\begin{IEEEeqnarray}{rCl} 
	\lxc &=& 	\begin{pmatrix}
				1 & -\rx & 0 & \ldots & -(1-\rx) \\
				-(1-\rx) & 1 & -\rx & \ldots & 0 \\
				\vdots & \vdots & \vdots &  \ddots & \vdots \\
				-\rx & 0 & \ldots & -(1-\rx) & 1
			\end{pmatrix}.
			\label{eq:lxc}
\end{IEEEeqnarray}
}
$\lvc$ is defined similarly by replacing $\rx$ by $\rv$ in
(\ref{eq:lxc}). 

For such a system, first the conditions for asymptotic stability are
found.
The most important condition is that $\rx=0.5$ \cite[Prop. 3.5]{Cantos2014a}---there must be a symmetric
coupling in the position. For stable circular systems it is shown in \cite[Thm.
4.8]{Cantos2014a} that an external input or a disturbance causes two signals to
propagate in the system in opposite directions and with different velocities.
These so-called signal velocities are calculated from
the phase velocities \cite[Lem. 4.4]{Cantos2014a}. We will use the same
ideas in this paper (Sec. \ref{subsec:sigVel}) to describe a stable system
in terms of traveling waves.

The paper \cite{Cantos2014b} studies transients in path systems. Its main result is the description of the
 transient in the path graph using two travelling waves, attenuated at the
 boundaries \cite[Thm. 3.5]{Cantos2014b}. 
 The connection of the circular system from \cite{Cantos2014a} to the path system (which is the one we are
really interested in) was conjectured in \cite{Cantos2014b} as follows:
the
asymptotic instability of the circular system should imply either flock or asymptotic instability
of the path system. We now restate flock stability for our system.
 \begin{defn}[Flock
stability, {\cite{Cantos2014b}}] The system (\ref{eq:overallSystemPath}) is
called flock stable if it is asymptotically stable and if \\$\max_{t \in
\mathbb{R}}|\pos_0(t) - \pos_N(t)|$ grows sub-exponentially in $N$ for the
conditions given in (\ref{eq:initCond}).
\end{defn}

 \subsection{Assumptions for solving}
The solution in \cite{Cantos2014b} is based on two main conjectures relating the
path and circular systems. Although we cannot prove them, we use them in the
present paper as well.
The final justification of both conjectures is in the
match of the predicted and simulated values, as shown in Sec. \ref{sec:simVer}.
The first one states that a local behavior of both systems is identical.
\begin{assum} 
If the path formation (\ref{eq:overallSystemPath}) is stable and flock stable,
then the behavior of a circular system is the same as in the path system
for vehicles reasonably far from the boundaries.
	\label{assum:localBeh} 
\end{assum}
This assumption allows to use properties derived for the circular graph (which
is much easier to analyze) in the path graph. We note that this treatment
implies that the boundary condition for the $N$th agent in the path graph, if
reasonable, does not enter the analysis, and therefore that boundary condition
does not affect our conclusions.

A similar assumption has been made by others (see \cite{Bamieh2002, Bamieh2012, 
D'Andrea2003}) to simplify the analysis and make the system spatially invariant.
In fact, in Solid State Physics this idea is known as \emph{periodic boundary conditions} and 
goes back to the beginning of the 20th century (see \cite{Ashcroft1976}).

In order to investigate flock stability in the path system using properties of a
circular system, we need an additional assumption about relations of the two
interconnections.
\begin{assum}
	If the circular formation is asymptotically
	unstable, then the path formation is either asymptotically unstable or flock
	unstable.
	\label{assum:flockUnstab}
\end{assum}
The explanation in \cite[Def. 3.2]{Cantos2014b} is that the path system has non-normal eigenspaces which makes the bad effects
more pronounced. To this explanation we can add an additional one based on the
travelling wave concept in distributed control \cite{Martinec2014b}.
Asymptotic instability can be caused by the travelling wave which is amplified as it travels in the formation.
The amplification will happen far from boundaries also in the path system.
This results either in an asymptotic instability or in a flock
instability (if the reflections at the boundaries attenuate the signal
sufficiently) of the path system. 

Both assumptions are illustrated for our third-order system in Fig.
\ref{fig:pathCircResp} which shows the initial responses ($\posdist_{35}(0)=2$
and $d_{i,j}=0\, \forall i,j$) of the circular and path system. As can be seen,
the signal gets amplified as it propagates from one agent to the other. On the other hand, individual agent's
response goes to zero, until the amplified travelling wave gets back to the
agent after propagating through all other agents (see the sharp growth at time $70$ in Fig.
\ref{fig:circResp}).
The initial behavior of the path graph and circular systems is the same --- both
amplify the signal.
\begin{figure}[t]
\centering
	\begin{subfigure}[b]{0.22\textwidth}
	\includegraphics[width=1\textwidth]{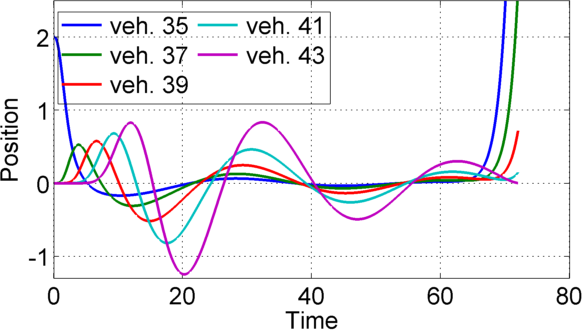}
	\caption{Circular system.}
	\label{fig:circResp}
	\end{subfigure} 
	\begin{subfigure}[b]{0.22\textwidth}
	\includegraphics[width=1\textwidth]{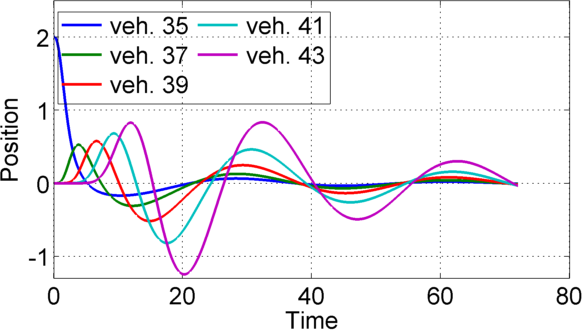}
	\caption{{Path system.}}
	\label{fig:pathResp}
	\end{subfigure}   
	\caption{Signal propagation in initial-condition response. a) response of an
	asymptotically unstable circular system, b) response of a flock unstable path
	system. In both cases $N=70, \fric = 3,\gx=2, \gv = 3, \rv=\rx=0.33$.}
	\label{fig:pathCircResp}
\end{figure}

\vspace{-10pt}
\section{Analysis of the circular system}
If $\rx \neq \rv$ in the path system in (\ref{eq:overallSystemPath}), there are
two different Laplacians $\lx$ and $\lv$ which are not simultaneously
diagonalizable. This prevents many convenient approaches to guarantee stability
such as a synchronization region \cite{Zhang2011} or LMI-based criterion
\cite{Massioni2009}. This makes the stability and performance
analysis of the path system very difficult.

To overcome this limitation, we invoke the Assumptions \ref{assum:localBeh} and
\ref{assum:flockUnstab} to extract some properties of the circular system and
apply them in the analysis of the path system.
To obtain such properties, we assume in this section that the
communication structure is the circular graph 
with Laplacians $\lxc, \lvc$. These Laplacians are circulant
matrices, which are \emph{simultaneously diagonalizable}.
Note that we investigate
the circular system only in order to learn something about the path system --
the circular system is \emph{not of interest} by itself. 

\subsection{Stability of the circular system}
When we assume the circular interaction topology, the
state-space model has a form
{\setlength{\arraycolsep}{2pt}
\begin{equation}
\dfrac{d}{dt}\! \begin{pmatrix} 
				\pos\\ 
				\dot \pos \\
 				\ddot \pos 
 			\end{pmatrix} \!=\!
\systMatCirc \!	\begin{pmatrix} 
				\pos \\ 
				\dot \pos \\ 
				\ddot \pos 
			\end{pmatrix} \!\equiv\!
			\begin{pmatrix} 
				0 & I & 0 \\
				0 & 0 & I \\
				-\gx \lxc & -\gv \lvc & -\fric I \end{pmatrix} \!
\begin{pmatrix} \pos\\ \dot \pos \\ \ddot \pos \end{pmatrix},
\label{eq:overallSystemCircle}
\end{equation}
}

The Laplacians $\lxc, \lvc$ are diagonalizable by the discrete Fourier
transform.
So let $\eigVectL_m$ be the $m$-th eigenvector of $\laplc$, that is the vector
whose $j$-th component satisfies
\begin{equation}
(\eigVectL_m)_j = e^{\imath \sfr j} \equiv e^{\imath\,\frac{2\pi m}{N+1}\,j},
\quad j=0,1,\ldots,N,
\end{equation}
with $\sfr=2\pi m/(N+1)$ and $\imath = \sqrt{-1}$. By \cite{Cantos2014b}, we
calculate the eigenvalues $\eiglx$ of $\gx \lx$ and $\eiglv$ of $\gv \lv$ as
\begin{equation}
\begin{array}{ccc}
\eiglx (\sfr) & = & \gx[1-\cos \sfr + \imath (1-2\rx) \sin \sfr],
\\
\eiglv(\sfr) & = &\gv[1-\cos \sfr + \imath (1-2\rv) \sin \sfr].
\end{array}
\end{equation}

Let us denote $\bx = 1-2\rx, \: \bv = 1-2\rv$. We can expand
the eigenvalues $\eiglx$ and $\eiglv$ in the Taylor series
\begin{IEEEeqnarray}{rCl}
 	\eiglx(\sfr) & = &
\gx \big[\imath \bx \sfr + \frac{1}{2} \sfr^2-\frac{\imath}{6}  \bx
\sfr^3 \ldots \big],
\label{eq:lxTaylorSer}
\\
\eiglv(\sfr) & = & \gv \big[\imath \bv \sfr+ \frac{1}{2} \sfr^2-\frac
{\imath}{6} \bv \sfr^3 \ldots \big].
\label{eq:lvTaylorSer}
\end{IEEEeqnarray}

We now calculate three eigenvalues $\eigm_{m,i}, \: i=1,2,3$ of $\systMatCirc$
associated with $\eigVectL_m$ for each $m$. The eigenvalue equation is
\begin{equation}
	\begin{pmatrix} 
		0 & I  & 0\\ 
		0 & 0 & I \\
		-\gx \lxc & -\gv \lvc &  -\fric I
	\end{pmatrix}
	\begin{pmatrix}
		 \eigVectM_{m,1}\\ 
		 \eigVectM_{m,2}\\
		 \eigVectM_{m,3}
	\end{pmatrix}
	= \eigm_{m,i}\,
	\begin{pmatrix} 
		\eigVectM_{m,1}\\
		\eigVectM_{m,2}\\
		\eigVectM_{m,3}
	\end{pmatrix},
\end{equation}
from which follows that $\eigVectM_{m,2}=\eigm_{m,i} \eigVectM_{m,1}, \quad
\eigVectM_{m,3} = \eigm_{m,i}\eigVectM_{m,2}$ and $-\gx \lxc \eigVectM_{m,1} 
-\gv \lvc \eigm_{m,i} \eigVectM_{m,1} -\fric I \eigm_{m,i}^2
\eigVectM_{m,1}=\eigm_{m,i}^3 \eigVectM_{m,1}$. Since $\lxc$ and $\lvc$ have
as an eigenvector $\eigVectL_m$, it follows from the last equality that
$\eigVectM_{m,1}=\eigVectL_m$. For simplicity of notation we drop the subscripts
of $\eigm$ except from when ambiguity seems possible.
From the last row it is elementary to get a set of
$N+1$ eigenvalue equations for $\systMatCirc$:
\begin{IEEEeqnarray}{rCl}
\eigm^3 + \fric \eigm^2 + \eiglv(\sfr)\eigm + \eiglx(\sfr)=0.
\label{eq:chareqNoPlug}
\end{IEEEeqnarray}
Substituting the expressions for $\eiglx(\sfr)$ and $\eiglv(\sfr)$, we
get
\begin{IEEEeqnarray}{rCl}
\eigm^3 +\fric \eigm^2+\gv[1-\cos \sfr + \imath (1-2\rv) \sin
\sfr]\eigm && \nonumber \\
+\gx[1-\cos \sfr + \imath (1-2\rx) \sin \sfr]&=&0.
\label{eq:characteristic3}
\end{IEEEeqnarray}
By letting
$\sfr$ equal $0$ or $\pi$ we get real polynomials
\begin{IEEEeqnarray}{rCl}
\eigm^3 + \fric\eigm^2 &=& 0, \label{eq:rootOpenLoop}\\
\eigm^3 + \fric \eigm^2 + 2 \gv\eigm + 2\gx &=& 0. \label{eq:realCoef2}
\end{IEEEeqnarray}
The equation (\ref{eq:realCoef2}) implies via Routh-Hurwitz criterion a simple
necessary conditions for stability.

\begin{lem} The necessary conditions for the stability of
(\ref{eq:overallSystemCircle}) for all $N$ are $\fric>0$, $\gx>0$ and $\gv> 0$
and $\fric> \gx/\gv$.
\label{lem:1.1-3rd-order}
\end{lem}

An important result is that a symmetric interaction in position is necessary
for the stability of the circular system.
\begin{lem} The necessary
condition for the stability of (\ref{eq:overallSystemCircle}) for all $N$ is
$\rx=1/2$, i.e., $\bx=0$.
\label{lem:x-symm-3rd-order}
\end{lem}
\begin{proof} From (\ref{eq:characteristic3}) for each value of $\sfr$
we obtain 3 roots $\nu(\sfr)$. Two of the roots in (\ref{eq:rootOpenLoop})
are zero.
We will find the behavior of nearby roots (for $|\sfr|$ small). Since roots of polynomials
are continuous functions of their coefficients, there will be
2 branches of roots with $\nu(\sfr)$ tending to zero as $\phi$ tends to zero
from above, and another 2 branches with $\nu(\sfr)$ tending to zero as $\sfr$ 
tends to zero from below.

Assume that $\beta_x\neq 0$. Expand (\ref{eq:characteristic3}) in terms of lowest
order (in $\nu^i\sfr^j$). We obtain
$\fric \eigm^2 + i \beta_x g_x \sfr=0.$
It is easily seen that for all values $|\sfr|$ this has two solutions in the 
right half-plane. It follows that if $|\sfr|$ is small enough,
(\ref{eq:characteristic3}) must have two solutions in the right half-plane. It
follows that we must choose $\beta_x= 0$, so $\rx=1/2$.
\end{proof}

The next theorem guarantees the stability of an arbitrary
large system with a circular topology.

 \begin{theo} All non-trivial eigenvalues of (\ref{eq:overallSystemCircle})
have negative real part if and only if all of the below hold:
\bsenn
\begin{array}{l}
\mathrm{I.}:\quad \fric > 0 \logand \gx > 0 \logand \gv > 0 \logand a>\gx / \gv,
\\
\mathrm{II.}: \quad \rx = 1/2,\\
\mathrm{III.}: \quad 1-2\rv \in \left(-\dfrac{\fric \gv - \gx}{\sqrt{2\gv^3}},
\dfrac{\fric \gv - \gx}{\sqrt{2\gv^3}}\right).
\end{array}
\esenn
\label{thm:stability-3rd-order}
\end{theo}
\vspace{-20pt}
 \begin{proof} 
 Let us call the statement ``Circular system (\ref{eq:overallSystemCircle}) is
 stable'' as S.
 The necessity of conditions I. and II. were proved above. We will use them to
 prove that given I. and II., then III. is false is equivalent to S is false.

We know from (\ref{eq:rootOpenLoop}) that $-a$ is a solution of
(\ref{eq:characteristic3}) and that it lies in the left half-plane. By
continuity of roots of polynomials all the solutions of
(\ref{eq:characteristic3}) must lie on a curve $\eigm(\sfr)$ starting at $-a$.
To have unstable roots on the curve $\eigm(\sfr)$, the curve must cross the imaginary axis for some $\sfr \in (0, 2\pi)$.
Then there must be purely imaginary solutions  $\imath \omega$ ($\omega$ real)
to (\ref{eq:characteristic3}).
Substitute $\imath \omega$ for $\eigm$ into (\ref{eq:characteristic3}) to get
\begin{equation}
-\imath \omega^3-\fric \omega^2+ [\imath(1\!-\!\cos\sfr)\!-\! \bv\sin\sfr]\gv
\omega + \gx(1 \!-\! \cos\sfr)=0. \label{eq:imagAxis}
\end{equation}
The real and imaginary parts of (\ref{eq:imagAxis}) are, respectively:
\begin{IEEEeqnarray}{rCl}
 -\fric \omega^2-\gv \bv \omega \sin\sfr  +\gx(1-\cos \sfr) &=& 0,
 \label{eq:imagAxisReal}\\
 -\omega(\omega^2-\gv (1-\cos \sfr))&=&0. \label{eq:imagAxisImag}
\end{IEEEeqnarray}
So $S$ is false if both of the last equations hold (i.e.,
(\ref{eq:imagAxis}) has a solution).
The equation (\ref{eq:imagAxisImag}) holds for $\omega=0$ or
$\omega^2=\gv(1-\cos\phi)$, where $\omega=0$ gives only the `trivial'
eigenvalue (namely $\phi=0$). Plugging the other solution
$\omega^2=\gv(1-\cos\phi)$ into (\ref{eq:imagAxisReal}) gives:
\begin{equation}
\bv=\pm
\dfrac{\fric \gv-\gx}{\sqrt{2 g_v^3}}\,\dfrac{\sqrt{2(1-\cos\sfr)}}{\sin\sfr}.
\end{equation}
The factor $\frac{\sqrt{2(1-\cos\phi)}}{\sin\phi}$ maps the unit
circle onto $[-\infty,-1]\cup[1,\infty]$. So for $|\bv|\geq
\frac{\fric \gv - \gx}{\sqrt{2|\gv|^3}}$ there exists $\sfr$ for which equation
(\ref{eq:imagAxis}) is satisfied, the system then has purely imaginary roots and
therefore the system can be unstable. If $|\bv|<
\frac{\fric \gv - \gx}{\sqrt{2|\gv|^3}}$, then no imaginary solution exists and
whole curve $\eigm(\sfr)$ lies in the stable half-plane.
\end{proof}

\subsection{Signal properties}
\label{subsec:sigVel}
 \begin{figure}[t]
\centering
	\includegraphics[width=0.35\textwidth]{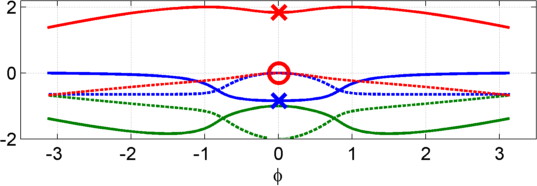}
	\caption{Phase velocities calculated by (\ref{eq:wvgen}) in $vehicles/s$ (solid)
	and the corresponding damping (dashed) as a function of $\sfr$. There are three
	waves for each $\sfr$, two with negative velocity (blue and green) and one with positive (red).
	The phase
	velocities in green have all high damping, so they do
	not affect the signal velocity.
	 The signal velocities from
	(\ref{eq:velc}) are shown by red ($\cp$) and blue ($\cm$) crosses, the
	corresponding damping by a red circle.
	  $N=500, \fric=2, \gx=6.2, \gv=10, \rv=0.4$.}
	\label{fig:sigVel}   
\end{figure}

Similarly to \cite{Cantos2014a}, we would like to obtain the signal velocity in
our circular system (\ref{eq:overallSystemCircle}). By \cite[Lem.
4.4]{Cantos2014a}, the phase velocity $\wv_{m, i}$ and its damping
$\alpha_{m,i}$ for a mode associated with a given $\sfr=2\pi m/(N+1)$ can be
calculated as
\begin{equation}
\wv_{m,i}=-\text{Im}(\eigm_{m,i})/\sfr, \:\:\: \alpha_{m,
i}=\text{Re}(\eigm_{m, i}), \,i=1,2,3. \label{eq:wvgen}
\end{equation}
We are interested in modes with
very low damping, since they travel in the system with slow decay --- they give
us the signal velocity. Thus, we want to find the eigenvalues $\eigm(\sfr)$ with
small real parts. They are those corresponding to $\sfr \rightarrow 0$. To find them, first expand $\eigm(\sfr)$
as 
\begin{IEEEeqnarray}{rCl}
	\eigm(\sfr) &=& \imath \eigmcoef_1\sfr+\frac{1}{2} \eigmcoef_2 \sfr^2
		+\frac{\imath}{6} \eigmcoef_3\sfr^3 \ldots ,
		\label{eq:eigmTaylorExp} 
\end{IEEEeqnarray}
We substitute the expansions (\ref{eq:lxTaylorSer}), (\ref{eq:lvTaylorSer}) and
(\ref{eq:eigmTaylorExp}) into (\ref{eq:chareqNoPlug}). Notice that the expansion
(\ref{eq:eigmTaylorExp}) works because the terms depending on $\sfr$ cancel in
(\ref{eq:characteristic3}).
We collect terms of order $\sfr^2$, and $\sfr^3$, etc. The coefficients of these
orders must be zero and that will determine $\eigmcoef_i$. The first
non-trivial equation is the coefficient of $\sfr^2$. It reads:
${\cal O}(2): \, a \eigmcoef_1^2+\gv \bv \eigmcoef_1-\frac{1}{2} \gx=0.$
We calculate $\eigmcoef_1$ as
$\eigmcoef_1 = \dfrac{-\gv \bv \pm \sqrt{2\fric \gx +\gv^2\bv^2}}{2\fric}.$

Since for $\sfr$ small by (\ref{eq:eigmTaylorExp}) and
(\ref{eq:wvgen}) $\text{Im}(\eigm(\sfr))\approx\eigmcoef_1 \sfr$, the
coefficient $-\eigmcoef_1$ determines the signal velocities.
\begin{lem}
	The signal velocities are given as 
	\begin{equation}
\wv_\pm=\dfrac{\gv \bv \pm \sqrt{\gv^2 \bv^2+2\fric \gx}}{2\fric},
\label{eq:velc}
\end{equation}
where $\cp>0$ and $\cm<0$ (velocity in vehicles/second).
\label{lem:signalVelocity} 
\end{lem}
By the stability conditions we mentioned, (\ref{eq:velc}) gives one
positive real and one negative real solutions (red and blue crosses in Fig.
\ref{fig:sigVel}). The wave with the positive velocity $\cp$ propagates in the
direction with growing vehicle index and the wave with $\cm$ the other way.

\section{Transients in the path system}
We have obtained enough properties of the circular system to derive the
transients the original \emph{path system} (\ref{eq:overallSystemPath}). The
transient we analyze is when the platoon is in steady-state and the leader starts to move with unit velocity (\ref{eq:initCond}).

We have Theorem
\ref{thm:stability-3rd-order} guaranteeing the stability of the circular system
which by Assumption \ref{assum:flockUnstab} allows for flock stability of the
path system. The signal velocity in (\ref{eq:velc}) should remain the
same in the path system --- Assumption \ref{assum:localBeh}. The boundary conditions are
the same as in \cite{Cantos2014b} --- the leader driven independently of the platoon and 
the agent $N$ having no follower. 
For stable systems (in both senses), the orbit of the last agent
can be characterized by the following quantities (see Fig. \ref{fig:responseOpt}):
half-period $\period$ is the smallest $t>0$ such that $\pos_\numVeh(t)-\pos_0(t)=0$ and the amplitude
$\ampl_i$ of the $i$th oscillation is $\ampl_i=\max_{t \in
[(i-1)\period, i\period]} |\pos_0-\pos_N|$. We can now restate
\cite[Thm. 3.5]{Cantos2014b} for the system with friction. 

\begin{theo} If Assumptions \ref{assum:localBeh} and \ref{assum:flockUnstab}
hold and the path system (\ref{eq:overallSystemPath}) is asymptotically stable,
the parameter values satisfy the conditions in Theorem \ref{thm:stability-3rd-order} and as $\numVeh$ tends to infinity, the system (\ref{eq:overallSystemPath}) will behave as a wave equation with
boundary conditions. The signal velocities are given by 
(\ref{eq:velc}).
 In particular, if from
an equilibrium position at rest, the leader starts to move with a unit velocity
at $t=0$, then the characteristics of the orbit of $\pos_0(t) - \pos_\numVeh(t)$ are:
\begin{IEEEeqnarray}{rCl}
\ampl_1&=&\dfrac{\numVeh}{|\cp|},  \label{eq:time}\\
 |\ampl_{k+1}/\ampl_{k}|&=& |\cm|/|\cp|,
\label{eq:amplDecay} \\
\period &=& \numVeh \left|
\frac{1}{|\cp|}+\frac{1}{|\cm|}\right|.
\label{eq:period}
\end{IEEEeqnarray}
\label{thm:velocities-3rd-order}
\end{theo}
 \begin{figure}[t]
\centering
	\includegraphics[width=0.42\textwidth]{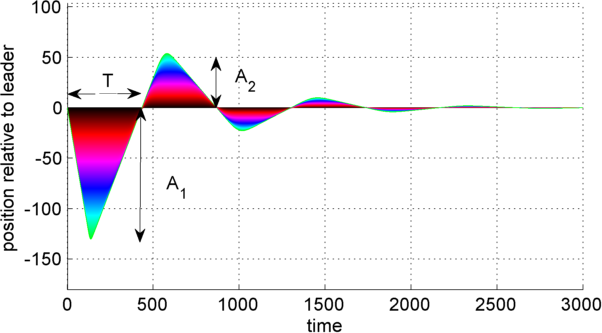}
	\caption{Spacing error to the leader $\pos_0-\pos_i$ with optimized
	controller (see Sec. \ref{sec:optimization} and
	\ref{sec:optimizationResults}). $N = 250, \fric=2, \gx = 6.2, \gv
	= 10$ and $\rv=0.4$.}
	\label{fig:responseOpt} 
\end{figure}
The proof of \cite[Thm. 3.5]{Cantos2014b} uses only boundary conditions and
wave velocities, hence it remains valid for our case as well.
Note that when the leader starts moving, this causes first a wave with velocity
$\cp$, which then reflects at agent $\numVeh$ as a wave with velocity $\cm$. 
Notice that we want $|\cm|/|\cp|$ to be less than 1 to avoid
exponential growth of the amplitudes. Since $\gv$ and $\fric$ must be positive, we want to keep
$\bv>0$, i.e., $\rv < 1/2$ and the agent pays more attention to the front
velocity error.
 
\vspace{-10pt}

\section{Simulation verification}
\label{sec:simVer}
Fig.~\ref{fig:relative_erros_comp} numerically validates
Theorem~\ref{thm:velocities-3rd-order} by calculating the relative error between
the predicted and measured values as a function of $\numVeh$. Let $\chi$
be a given quantity of interest in Thm. \ref{thm:velocities-3rd-order} --- either $\ampl_i$,
$\ampl_{i+1}/\ampl_{i}$ or $T$.
Let $\chi_{\text{pred}}$ be the value predicted by (\ref{eq:time}),
(\ref{eq:amplDecay}) or (\ref{eq:period}), respectively, and
$\chi_{\text{meas}}$ be the value measured from the numerical simulations of a finite platoon. The error is calculated as
\begin{align}
  \vartheta = \log\left(\frac{\chi_{\text{pred}}}{\chi_{\text{meas}}}-1\right), \label{eq:relative_error_verif}
\end{align}
We can see
that the relative error of each predicted parameter decreases exponentially with
the increasing number of vehicles in the platoon. This confirms the asymptotic
formulas in Theorem \ref{thm:velocities-3rd-order} and also Assumptions
\ref{assum:localBeh} and \ref{assum:flockUnstab}.
\begin{figure}[ht]
 \centering
  \includegraphics[width=0.40\textwidth]{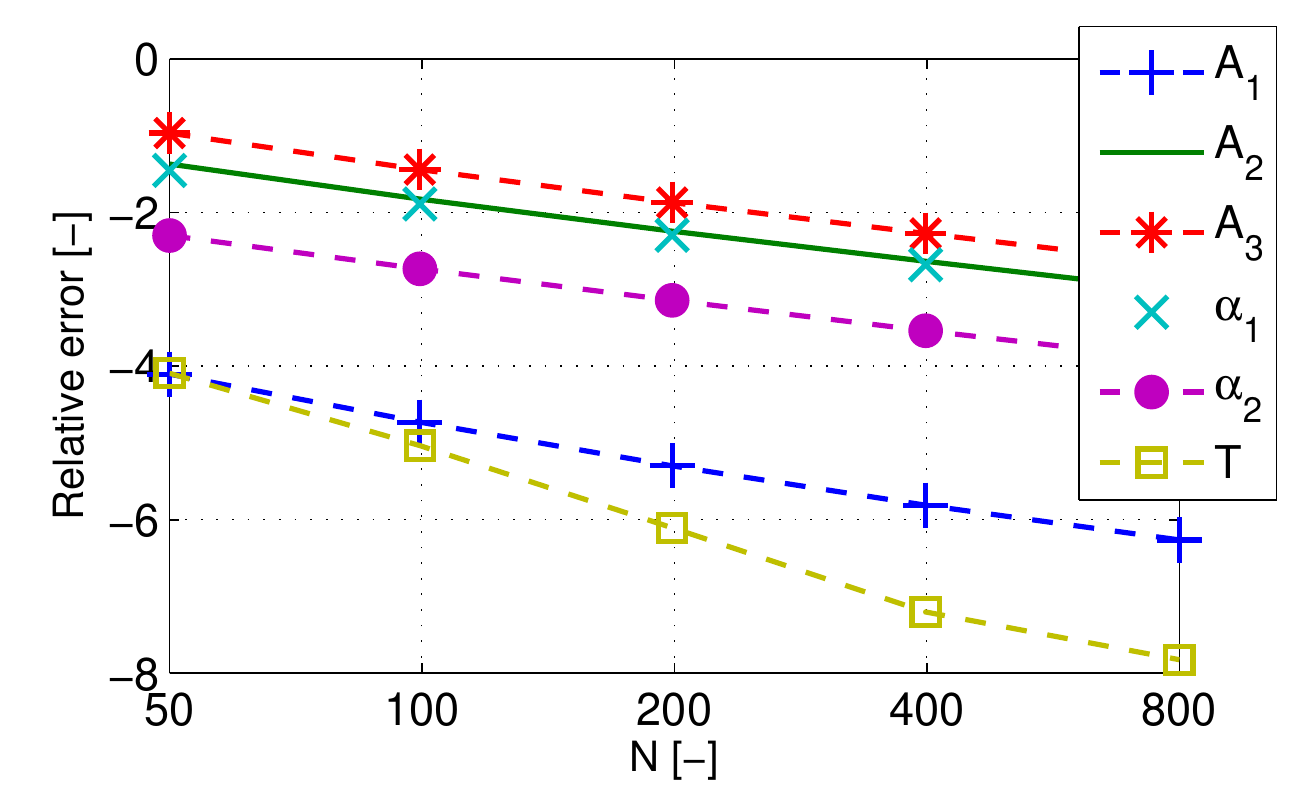}
  \caption{The numerical verification of (\ref{eq:time}), (\ref{eq:amplDecay})
  and (\ref{eq:period}). The relative errors are carried
  out with (\ref{eq:relative_error_verif}) for $g_x = 6.2$, $g_v = 10$, $a=2$, $\rho_x = 0.5$ and $\rho_v = 0.4$. The $\alpha_1$ and $\alpha_2$ are the attenuation coefficients from (\ref{eq:amplDecay}) for $A_2/A_1$ and $A_3/A_2$, respectively.}
  \label{fig:relative_erros_comp}
\end{figure}
 
The numerical simulations in Fig.~\ref{fig:strategies_comp} show that a platoon
with a controller tuned symmetrically (the left panel) has a very long
transient. The transient is shortened for the case of the asymmetric controller
(the middle panel), however, the overshoot of such a platoon is extremely large,
which is a consequence of Lemma~\ref{lem:x-symm-3rd-order} --- the circular
system is unstable, hence the path system is flock unstable.
When we set the asymmetry only in the velocity (right), then both the transient and
the overshoot are reasonable.
\begin{figure*}[ht]
 \centering
  \includegraphics[width=0.75\textwidth]{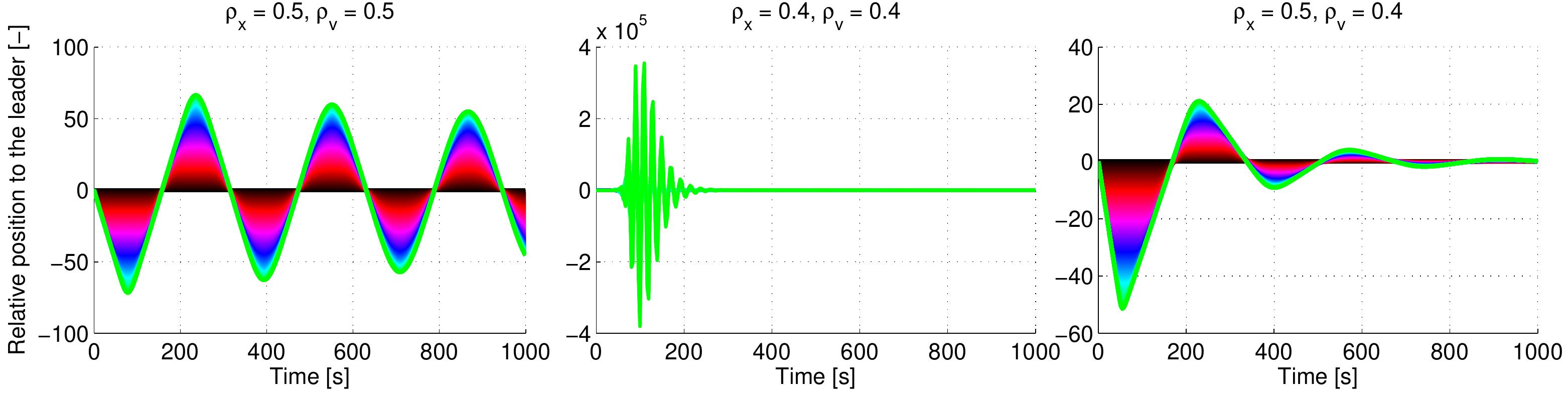}
  \caption{The numerical simulations comparing the responses of three different
  control strategies for $N=100$, when the leader changes its velocity from 0 to
  1. The figure shows the relative positions to the leader
  $\pos_0(t)-\pos_i(t)$ of all the vehicles for three different combinations of
  $\rho_x$ and $\rho_v$. For all three cases were $g_x = 6.2$, $g_v = 10$ and
  $a=2$.}
  \label{fig:strategies_comp}
\end{figure*}

%

\section{Optimization of system parameters}
\label{sec:optimization}
The previous section gave us signal velocities and amplitudes of the transient,
which depend on the gains $\gx$ and $\gv$ and the velocity asymmetry $\rv$. In
this section we give an approach how to select these three parameters. We assume that
the friction $\fric$ is given by the vehicle model and cannot be affected by the
designer. 

We propose the following method for ``optimal'' (due to asymptotic formulas)
gain and asymmetry selection.
It is based on minimizing the absolute value of the spacing error of all
vehicles in the formation, denoted as $\areaosc$, when the leader starts to move
from the stand-still. Therefore, the optimization has a form
\begin{equation}
	\min_{\gx, \gv, \bv} \areaosc = \min \sum_{i=1}^{\numVeh} \int_{0}^{\infty}
	|\erl_i(t)| dt,
	\label{eq:totalAbsError}
\end{equation}
where the error is given by $\erl_i(t)=\pos_0(t)-\pos_i(t)$. Clearly, $\erl_0
= 0$. $\areaosc$ is minimized over $\gx, \gv, \bv$. 

Theorem \ref{thm:velocities-3rd-order} tells us that the system behaves as a
wave equation with boundaries. After a unit change of leader's velocity, first
the signal spreads from the leader to vehicle $\numVeh$ with velocity $\cp$ and then
it reflects back with velocity $\cm$. The graph of the response of the last vehicle in the formation must then be almost triangular, as shown in Fig.
\ref{fig:responseOpt}. The error of the first oscillation for the last
vehicle (before $\erl_\numVeh$ gets back to zero for the first time) is
\begin{equation}
	\areaosc_{N,1}=\int_0^{T} |\erl_\numVeh(t)| dt \approx \frac{1}{2} \period
	\ampl_1.
\end{equation}

To get the error of the other
agents we assume that the maximal value of the error of $i$th agent is given
by $\frac{\ampl}{\numVeh}i$ (the peaks are uniformly spaced from 0 to $A_1$ with
distance $\frac{\ampl_1}{\numVeh}$). Then the shape of the error is almost a
trapezoid with one base of length $\period$ and the other with 
$(\numVeh-i)\left(\frac{1}{|\cp|}+\frac{1}{|\cm|}\right) =
\period-\frac{i}{\numVeh} \period = \period \left(1-\frac{i}{\numVeh} \right).$
The absolute value of the error of the $i$th vehicle in the first
oscillation is approximately the area of the trapezoid
\begin{IEEEeqnarray}{rCl}
	\areaosc_{i,1}&=&\int_0^{\period} |\erl_i(t)| \approx  
	 \frac{\ampl_1}{\numVeh}i
	\period \left(1-\frac{i}{2 \numVeh}\right).
\end{IEEEeqnarray}
We have approximated the first oscillations. The errors of the others
are calculated in a similar same way, i.e. the period is again $\period$ and the
amplitude is obtained using (\ref{eq:amplDecay}). The total absolute value of the error of
the $i$th agent is 
	$\areaosc_i = \int_{0}^{\infty} |\erl_i(t)| dt,$
which is using the trapezoidal approximation given as the sum of areas of
all oscillations
\begin{IEEEeqnarray}{rCl}
	\areaosc_i &=& \int_0^\infty |\erl_i(t)| dt = \sum_{j=1}^{\infty}
	\areaosc_{i,j} \approx \sum_{j=1}^{\infty} \frac{\ampl_j}{\numVeh}i \period \left(1-\frac{i}{2 \numVeh}\right) \nonumber
	\\
	&=& 
	\frac{\ampl_1}{\numVeh} i
	\period \left(1-\frac{i}{2 \numVeh}\right) \frac{1}{1-\frac{|\cm|}{|\cp|}}.
	\label{eq:trapApproxI}
\end{IEEEeqnarray}
We used (\ref{eq:amplDecay}) to quantify the amplitude of the $j$th oscillation
and then the sum of geometric series since $\frac{|\cm|}{|\cp|} < 1$.

Our criterion (\ref{eq:totalAbsError}) captures the sum of $\areaosc_i$ of
all agents. It can be calculated as
\begin{IEEEeqnarray}{rCl}
	\areaosc &=& \sum_{i=1}^\numVeh \areaosc_i \approx \sum_{i=1}^\numVeh
	\frac{\ampl_1}{\numVeh} i \period \left(1-\frac{i}{2 \numVeh}\right)
	\frac{1}{1-\frac{|\cm|}{|\cp|}} \nonumber \\
	 &=& \ampl_1
	\frac{1}{1-\frac{|\cm|}{|\cp|}} \period C = J C.
	\label{eq:optimCrit0}
\end{IEEEeqnarray}
with $C = \sum_{i=1}^\numVeh  \frac{i}{\numVeh} \left(1-\frac{i}{2
\numVeh}\right)$ being a constant which cannot be changed by optimization. Thus, it suffices to minimize $J$.
After plugging from (\ref{eq:period}) and (\ref{eq:amplDecay}), it has a form
\begin{IEEEeqnarray}{rCl}
	J &=& \frac{ \ampl_1 \period}{1-\frac{|\cm|}{|\cp|}}
	= 
	\left(\frac{|\cm|+|\cp|}{|\cp||\cm|}\right) \frac{\numVeh^2}{|\cp|-|\cm|} =
	\hat{J} \numVeh^2.
	\label{eq:optimCrit1}
\end{IEEEeqnarray}
 The number of agents is not part of the optimization and does not affect the
minimum. Plugging for the signal velocities from (\ref{eq:velc})  we evaluate
the sums and products as $
	|\cm|\!+\!|\cp|\! =\! (\sqrt{\gv^2 \bv^2+2 \fric \gx})/\fric, \:
	|\cm|\!-\!|\cp|\!=\!(\gv \bv)/\fric, \:
	|\cm||\cp| = \gx/(2 \fric)$.
With these terms the criterion (\ref{eq:optimCrit1}) becomes
\begin{equation}
	\hat{J} = \frac{\sqrt{\gv^2 \bv^2+2 \fric \gx}}{\gv \bv}\frac{2 \fric}{\gx} =
	\sqrt{\frac{1}{\gx^2}+\frac{2 \fric}{\gv^2 \bv^2 \gx }} 2 \fric.
	\label{eq:jhat}
\end{equation}
Since $\fric$ is a given constant and the square root is monotonic function, we
get the final optimization problem
\begin{IEEEeqnarray}{rCl}
	\min  \quad && \frac{1}{\gx^2}+\frac{2 \fric}{\gv^2 \bv^2 \gx }, 
	\label{eq:finalCrit}\\
	&&\text {s.t. conditions in Theorem \ref{thm:stability-3rd-order}.}
	\nonumber
\end{IEEEeqnarray}
The final
criterion is a function only of the parameters $\gx, \gv$ and $\bv$, which the
platoon designer can affect.

\subsection{Scaling of the absolute error}
The total error (\ref{eq:totalAbsError}) can be written using
(\ref{eq:optimCrit0}) and (\ref{eq:optimCrit1}) 
\begin{equation}
	\areaosc \approx J C = \hat{J} N^2 \sum_{i=1}^\numVeh  \frac{i}{\numVeh}
	\left(1-\frac{i}{2 \numVeh}\right). \label{eq:totalErrInN1}
\end{equation}
\begin{lem}
	The error $\areaosc$ in (\ref{eq:totalAbsError}) scales cubically with $N$ as
	\begin{equation}
		\areaosc(N) \approx \frac{\hat{J}}{12} N (N+1)(4N-1).
		\label{eq:totalErrEst}
	\end{equation}
\end{lem}
\begin{proof} 
The proof is a simple manipulation of (\ref{eq:totalErrInN1}). 
\begin{IEEEeqnarray}{rCl}
	\areaosc &\approx& \hat{J} N^2 \sum_{i=1}^\numVeh  \frac{i}{\numVeh}
	\left(1-\frac{i}{2 \numVeh}\right) = \hat{J} N^2 \sum_{i=1}^\numVeh  \frac{2 N
	i - i^2}{2\numVeh^2} \nonumber \\
	&=& \frac{\hat{J}}{2}\left[2N
	\sum_{i=1}^\numVeh i - \sum_{i=1}^\numVeh i^2\right] =
 \frac{\hat{J}}{2} \left(2N\frac{N(N+1)}{2} \right.\\
 && \left. - \frac{N(N+1)(2N+1)}{6} \right) =
 \frac{\hat{J}}{12} N (N+1)(4N-1). \nonumber \qedhere
\end{IEEEeqnarray} 
\end{proof}
The scaling of $\areaosc$  of different architectures with $N$ calculated from
simulations is in Fig.
\ref{fig:scaling}. It is clear that the error is the smallest for asymmetry only
in velocity and also that the error of asymmetric control with identical
asymmetries scales exponentially in $N$, which confirms flock instability ($\rx
\neq 0.5$). Also the predicted value matches the calculated one.
\begin{figure}[t]
\centering  
	\includegraphics[width=0.32\textwidth]{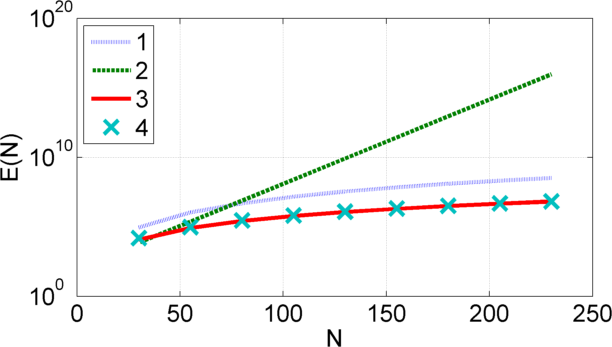}
	\caption{Logarithm of $\areaosc(\numVeh)$ for different architectures ---
	symmetric with $\rx=\rv=0.5$ (line 1), asymmetric with $\rx=\rv=0.4$ (line 2) and asymmetric
	with $\rx=0.5, \rv = 0.4$ (line 3). Line 4 shows estimate using
	(\ref{eq:totalErrEst}). Other parameters were $\gx=6.2, \gv=10, \fric=2$.}
	\label{fig:scaling} 
\end{figure}

\section{Optimization results}
\label{sec:optimizationResults}
Although both the criterion (\ref{eq:finalCrit}) and
the stability constraints are nonconvex, the optimization
using function \emph{fmincon} in Matlab for nonconvex optimization terminated quickly and successfully. 
The code used for simulations in the whole paper can be obtained at
\cite{Herman2015e}.


As follows from (\ref{eq:finalCrit}), the optimization procedure tried to
increase the gains $\gx$ and $\gv$  to decrease the
criterion. Therefore we specified upper bounds on $\gx, \gv$ to limit the
controller effort. The optimization was conducted for a given
friction $\fric=2$ and we got the values 
	$\gx = 6.2, \gv = 10$ and $\rv=0.4$.
The upper bounds for both gains $\gx, \gv$ were set to 10. To stay away from
the flock stability boundary, we changed the flock stability criterion in Thm.
\ref{thm:stability-3rd-order} to
\begin{equation}
	\bv = 1-2\rv \in \left(-\dfrac{\fric \gv - \gx}{\sqrt{2\gv^3}}+\varepsilon,
\dfrac{\fric \gv - \gx}{\sqrt{2\gv^3}}-\varepsilon \right),
\label{eq:modStabCrit}
\end{equation}
with $\varepsilon=0.1$.
The response is shown in Fig.
\ref{fig:responseOpt} and \ref{fig:strategies_comp}c.

\subsection{Robustness evaluation}
The optimization results should be verified to give robust results. The
simplest way to achieve robustness is to add some nonzero term $\varepsilon$ to
each of the stability criteria in Thm. \ref{thm:stability-3rd-order}, similarly
to (\ref{eq:modStabCrit}).
Then the system is not allowed to operate on the flock stability boundary.

The most important parameter of the system is the friction $\fric$. This might
change during the operation of the system and also might not be exactly known
apriori.
Using the values $\gx = 6.2, \gv = 10,\rv=0.4, N=1200$, we simulated the
response of the system for friction range $\fric \in [1.4, 2.8]$
and calculated the norm of the error using (\ref{eq:totalAbsError}).
Fig. \ref{fig:sensitivityToFric} shows how the norm changes with 
friction. We see that the change is approximately linear in friction and the
system has a good performance for a wide range of $\fric$.
The sharp growth for low friction caused by a flock instability confirms the
Assumption \ref{assum:flockUnstab}. The stability criterion of a circular system
in Thm. \ref{thm:stability-3rd-order} is violated for $a\leq 1.514$ since $\bv =
1-2\rv > \left( \frac{a \gv-\gx}{\sqrt{2|\gv|^3}}\right)$, making the system flock
unstable. As $N \rightarrow \infty$, the sharp growth appears exactly at the
critical point $a=1.514$.

 \begin{figure}[t]
\centering 
	\begin{subfigure}[b]{0.23\textwidth}
	\includegraphics[width=1\textwidth]{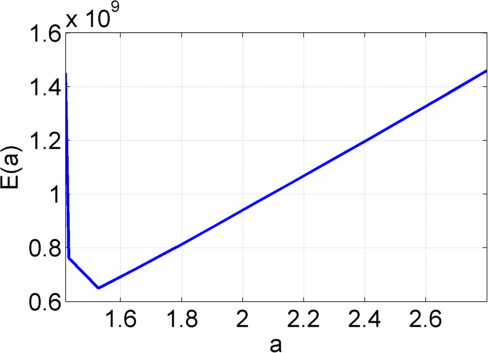}
	\caption{Sens. to change in $\fric$. }
	\label{fig:sensitivityToFric}  
	\end{subfigure}  
	\begin{subfigure}[b]{0.23\textwidth}
	\includegraphics[width=1\textwidth]{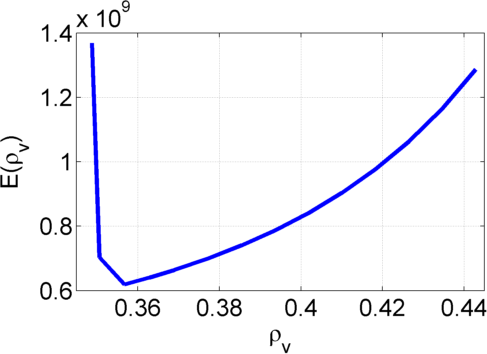}
	\caption{Sens. to change in $\rv$.}
	\label{fig:sensToRho}
	\end{subfigure} 
	\caption{Sensitivities to different parameters. The plots show value of
	$\areaosc$ of (\ref{eq:totalAbsError}) calculated from simulations.}
\end{figure}   
  
The Figure \ref{fig:sensToRho} shows how the error norm changes as a function of
$\rv \in [0.34, 0.44]$ with the optimal value $\rv^*=0.37$ (other parameters are
$\gx = 8.3, \gv = 10$, $\fric=2, N = 1200$). These are the values with
$\varepsilon=0$ in (\ref{eq:modStabCrit}). It is clear that the value $\rv^*$ is
almost the minimum of the function $\areaosc(\rv)$.
The better performance for lower (non-optimal) $\rv \leq \rv^*$ is due to
asymptotic formulas used. When number of vehicles increases, the minimum  will get closer to $\rv^*$.
Due to the sharp growth for $\rv < 0.36$ we recommend using $\rv=0.40$ to
achieve robustness, as obtained using (\ref{eq:modStabCrit}) with $\varepsilon=0.1$.
Then we get also $\gx=6.2$, as above.
\vspace{-10pt}
\section{Conclusion}
We investigated the transients in a vehicular platoon with identical vehicles
having friction. The asymmetries for a coupling in position and velocity are
different.
For some steps in the analysis, we considered circular system in order
to infer properties of the path system.
  
To achieve flock
stability, the coupling in the position must be symmetric. Then the path system
behaves as a wave equation with boundaries, where the travelling waves have
different velocities.
Using the velocities, we developed an optimization procedure to determine
the parameters of the controller. We conclude that the
behavior of a platoon with different asymmetries is superior to both symmetric
and identically asymmetric platoon.
 
We believe that the approach and some of the results shown here easily
generalize to more complicated systems having two integrators in the open loop.

\bibliography{CollaborativeControl-AsymmetryThirdOrder}

\vspace{\fill}
\end{document}